\documentclass[twocolumn,a4paper,aps]{revtex4}

\usepackage[english]{babel}
\usepackage[T1]{fontenc}
\usepackage[utf8]{inputenc}
\usepackage{bibrevtex4}

\usepackage{standalone}
\usepackage{fullpage}
\usepackage{microtype}
\usepackage{mathtools}
\usepackage{amssymb}
\usepackage{graphicx}
\usepackage[hidelinks]{hyperref}
\usepackage{tikz}
\usepackage{pgfplots}
\usepackage{booktabs}

\usepackage{csquotes}

\usepackage[binary-units=true]{siunitx}
\usepackage{lmodern}
\usepackage{paralist}
\usepackage[noabbrev,capitalise,poorman]{cleveref}
\usepackage{tipa}
\usepackage{textcomp}

\usepackage{enumerate}
\usepackage[symbol]{footmisc}
\usepackage{makeidx}
\usepackage{ntheorem}

\pgfplotsset{compat=1.9}
\usetikzlibrary{patterns}
\newtheorem{theorem}{Theorem}
\newtheorem{lemma}{Lemma}
\newtheorem{proof}{Proof}

\crefname{inequality}{Inequality}{Inequalities}
\creflabelformat{inequality}{#2{\upshape(#1)}#3}

\DefineBibliographyStrings{english}{%
    urlseen = {accessed},
}
\newcommand\defeq{\stackrel{\mathclap{\normalfont\tiny\mbox{def}}}{=}}

\begin{document}
\title{Tight Bounds for the Pearle-Braunstein-Caves Chained Inequality Without
the Fair-Coincidence Assumption}

\author{Jonathan Jogenfors}
\affiliation{Institutionen för systemteknik, Linköpings Universitet, 581 83
Linköping, Sweden}

\author{Jan-Åke Larsson}
\affiliation{Institutionen för systemteknik, Linköpings Universitet, 581 83
Linköping, Sweden}

\begin{abstract}
    In any Bell test, loopholes can cause issues in the interpretation of
    the results, since an apparent violation of the inequality may not
    correspond to a violation of local realism. An important example is the
    coincidence-time loophole that arises when detector settings might influence
    the time when detection will occur. This effect can be observed in many
    experiments where measurement outcomes are to be compared between remote
    stations because the interpretation of an ostensible Bell violation strongly
    depends on the method used to decide coincidence. The coincidence-time
    loophole has previously been studied for the Clauser-Horne-Shimony-Holt
    (CHSH) and Clauser-Horne (CH) inequalities, but recent experiments have
    shown the need for a generalization. Here, we study the generalized
    ``chained'' inequality by Pearle-Braunstein-Caves (PBC) with $N\geq 2$
    settings per observer. This inequality has applications in, for instance,
    Quantum Key Distribution where it has been used to re-establish security. In
    this paper we give the minimum coincidence probability for the PBC
    inequality for all $N\geq 2$ and show that this bound is tight for a
    violation free of the fair-coincidence assumption. Thus, if an experiment
    has a coincidence probability exceeding the critical value derived here, the
    coincidence-time loophole is eliminated.
\end{abstract}
\maketitle
\section{Introduction}
In recent years there has been an increased interest in the \enquote{chained}
generalization by Pearle, Braunstein and
Caves (PBC)~\cite{Pearle1970Hidden,Braunstein1990Wringing} of the
CHSH~\cite{Clauser1969Proposed,Bell1964Einstein} inequality due to its
applications in re-establishing a full Bell violation. An important application
is Quantum Key Distribution (QKD) based on the Franson
interferometer~\cite{Franson1989Bell} where it is
known~\cite{Aerts1999Twophoton,Jogenfors2014Energy,Jogenfors2015Hacking} that
the CHSH inequality is insufficient as a security test. If the switch to the
full PBC is made, full security can be
re-established~\cite{Jogenfors2014Energy,Tomasin2017Highvisibility}.

Where the standard CHSH inequality is limited to two possible measurement
settings per observer, the PBC inequality generalizes this to $N\geq 2$
settings. In order for Franson-based systems to function, $N\geq 3$ is required
at the cost of significantly higher experimental requirements. Specifically,
such an experiment requires a very high visibility, and until recently it was
believed~\cite{Jogenfors2014Energy} that these requirements were too impractical to
achieve. Recent works~\cite{Tomasin2017Highvisibility}, however, showed it possible
to meet these requirements by reaching a full violation of the PBC inequality for
$N=3$, $4$, and $5$ with visibility in excess of \SI{94.63}{\percent}.

Compared to other types of QKD such as BB84~\cite{Bennett1984Quantum} and
E91~\cite{Ekert1991Quantum}, the Franson design promises a simpler approach with
fewer moving parts. This advantage could allow the Franson system to pave the
way for commercial applications and widespread QKD adoption by reducing end-user
complexity~\cite{Jogenfors2015Hacking}. Therefore, the possibility of
re-establishing full security in the Franson interferometer is a strong
motivation of further study of the PBC inequality.

Previous works~\cite{Larsson2004Bell,Larsson2014Bell} have shown that the CHSH
and CH inequalities are vulnerable to the coincidence-time loophole which
relates to the problem of attributing detector clicks to the correct pair of
events. Bipartite Bell experiments measure correlations of outcomes between
remote stations, and as this is done for each \emph{pair} of detections, one
must reliably decide which detector clicks correspond to which pair. This is
more difficult than it might first appear due to high levels of non-detections,
jitter in detection times, and dark counts. If coincidences are lost, one needs
to apply the \enquote{fair-coincidence} assumption~\cite{Larsson2014Bell}, i.e.
that the outcome statistics is not skewed from these losses. According
to~\cite{Larsson2014Bell}, this fair-coincidence assumption appears to have been
implicitly made in at least every experiment before 2015.

This paper formally derives bounds for the coincidence probability so that a
violation of the PBC inequality can be performed without the fair-coincidence
assumption. Therefore, if the coincidence probability is high enough we can
eliminate the coincidence-time loophole. It should be noted that switching to
the generalized PBC inequality comes at a cost. As shown
by~\cite{Cabello2009Minimum}, the minimum required detection efficiency is
strictly increasing with $N$. Similarly, the PBC inequality in general has
higher requirements for the coincidence probability than the CHSH inequality.

We begin by formally defining the coincidence probability for PBC-based
experiments, followed by a sufficient condition for eliminating the
coincidence-time loophole. Then, we show
that our bound is tight by constructing a classical model that precisely
reproduces the output statistics whenever the losses exceed the bound. Finally,
we conclude that our results reduce to the special case of
CHSH~\cite{Larsson2004Bell} by choosing $N=2$ and compare with the
corresponding limits on detection efficiency~\cite{Cabello2009Minimum}.

\section{The coincidence-time loophole}

We use the symbol $\lambda$ for the hidden variable, which can take values in a
sample space $\Lambda$, that in turn is the domain of random variables
$A(\lambda)$ and $B(\lambda)$ denoting the measurement outcomes at Alice's and
Bob's measurement stations, respectively. We further assume that the space
$\Lambda$ has a probability measure $P$ which induces an expectation value $E$
in the standard way. We now give the formal definition of the PBC
inequality~\cite{Pearle1970Hidden,Braunstein1990Wringing}:
\begin{theorem}[Pearle-Braunstein-Caves]\label{thm:pbc}
    Let $N$ be an integer $\geq 2$ and $i$, $j$, and $k$ be integers between $1$
    and $2N$, and assume the following three prerequisites to hold almost
    everywhere:
    \begin{enumerate}[(i)]
        \item Realism: Measurement results can be described by probability
            theory, using two families of random variables $A_{i,j}$,$B_{i,j}$,
            e.g.,
            \begin{equation}
                \begin{split}
                    A_{i,j}: \Lambda \to & V \\
                    \lambda \mapsto & A_{i,j}(\lambda)\\
                    B_{i,j}: \Lambda \to & V \\
                    \lambda \mapsto & B_{i,j}(\lambda)\\
                \end{split}
            \end{equation}
        \item Locality: A measurement result should be independent of the remote
            setting, e.g., for $k\neq i$, $l\neq j$ we have
            \begin{equation}
                \begin{split}
                    A_{i,j}(\lambda)=& A_{i,l}(\lambda)\\
                    B_{i,j}(\lambda)=& B_{k,j}(\lambda)
                \end{split}
            \end{equation}
        \item Measurement result restriction: The results may only range from
            $-1$ to $+1$,
            \begin{equation}
                V=\{x\in \mathbb R; -1\leq x\leq +1 \}.
            \end{equation}
    \end{enumerate}
    Then, by defining
    \begin{equation}\label{eqn:pbc-value}
        \begin{split}
            S_{N}\defeq & \big|E(A_1B_1)+E(A_2B_1)\big|\\
            +&\big|E(A_2B_2)+E(A_3B_2)\big|+\cdots\\
            +&\big|E(A_{N}B_{N})-E(A_{1}B_{N})\big|
        \end{split}
    \end{equation}
    we get
    \begin{equation}\label[inequality]{eqn:pbc}
        \begin{split}
            S_{N}\leq 2N-2
        \end{split}
    \end{equation}
\end{theorem}
The proof consists of simple algebraic manipulations, adding of integrals and an
application of the triangle
inequality~\cite{Pearle1970Hidden,Braunstein1990Wringing}.

The right-hand value of \cref{eqn:pbc} is the highest value $S_N$
can attain with a local realist model. Compare this with the prediction of
quantum mechanics~\cite{Pearle1970Hidden,Braunstein1990Wringing}:
\begin{equation}\label{eqn:pbc-qm}
    S_N=2N\cos\left(\frac{\pi}{2N}\right).
\end{equation}
Note that $2N\cos\left(\pi/2N\right)>2N-2$ which, in the spirit of
\textcite{Bell1964Einstein}, shows that the outcomes of a quantum-mechanical
experiment cannot be explained in local realist terms.

Computing the Bell value requires computing the correlation between outcomes at
remote stations. Importantly, data must be gathered in pairs, so that products
such as $A_1B_1$ can be computed (see \cref{eqn:pbc-value}). Experimentally,
this is done by letting a source device generate pairs of (possibly entangled)
particles that are sent to Alice and Bob for measurement. Detectors at either
end record the measurement outcomes, and as previously mentioned there will is
always be variations on the detection times due to experimental effects. As a
consequence, it is not always obvious which detector clicks correspond to which
pairs of particles.

After a number of trials, Alice and Bob must determine in which trial if they have
coincidence (simultaneous clicks at Alice and Bob), a single event (only one
party gets a detection) or no detection at all. This is especially pronounced
if the experimental setup uses down-conversion where a continuous-wave
laser pumps a nonlinear crystal in order to spontaneously create pairs of
entangled photons. In that case the emission time is uniformly distributed over
the duration of the experiment so it becomes a probabilistic process that
further complicates pair
detection.

A typical strategy used in quantum optics experiments to reduce the influence of
noise in a Bell experiment is to have a time window of size $\Delta T$ around,
for example, Alice's detection
event~\cite{Giustina2013Bell,Christensen2015Analysis}. If a detection event has
occurred at Bob's side within this window it is counted as a coincident pair.
This is a non-local strategy as it involves comparing data between remote
stations and is used in many experiments.

For the experimenter it is tempting to choose a small $\Delta T$ since it
filters out noise and therefore increases the measured Bell value. At this
point, there is apparently no immediately obvious drawback of picking a very
small $\Delta T$. However, rejecting experimental data in a Bell experiment
modifies the underlying statistical ensemble and it is
known~\cite{Larsson1998Bell} to lead to inflated Bell values and a false
violation of the Bell inequality. This is a so-called loophole that can arise in
Bell testing, and many such loopholes have been studied in recent years
(see~\cite{Larsson2014Loopholes} for a review).

A coincidence window that is too small discards some truly coincident events as
noise. Therefore, the Bell value measurement only occurs on a subset of the
statistical ensemble which means a number of events are not accounted for when
the Bell value is computed. While the Bell value $S_N$ from \cref{thm:pbc} is in
violation of a Bell inequality, this violation might be a mirage. Specifically,
the loophole that arises from choosing a small $\Delta T$ is called the
\emph{coincidence-time loophole} and has previously been
studied~\cite{Larsson2004Bell} for the special CHSH case $N=2$. In addition,
more recent works~\cite{Larsson2014Bell} derive similar bounds for the
CH~\cite{Clauser1974Experimental} inequality.

We generalize the results previously obtained for the special case of CHSH
by deriving tight bounds for the coincidence-time
loophole in the PBC inequality (\cref{thm:pbc}) for all $N\geq 2$. This
contribution will be useful for future experiments investigating, among others,
Franson-based QKD\@. Other works~\cite{Cabello2009Minimum} have studied the
effects of reduced detector efficiency for the full PBC inequality, which in
turn is a generalization of an older result~\cite{Larsson1998Bell} that only
discussed the special CHSH case.

For the rest of this paper, Alice and Bob perform measurements on some
underlying, possibly quantum, system. Their measurements are chosen from
$\{A_i\}$ and $\{B_j\}$, respectively, i.e.\ sets of $N$ measurement settings
each. As
discussed by \textcite{Larsson2004Bell}, Alice's and Bob's choice of measurement
settings might influence whether an event is coincident or not. Following the
formalism in~\cite{Larsson1998Bell} we will therefore model non-coincident
settings $\lambda$ as subsets of $\Lambda$ where the random variables
$A_i(\lambda)$ and $B_i(\lambda)$ are undefined. We must therefore modify the
expectation values in \cref{eqn:pbc-value} to be conditioned on
coincidence in order for $S_N$ to be well-defined (see \cref{eqn:condition}).
The time of arrival at Alice's and Bob's measurement stations is defined as
\begin{equation}\label{eqn:time}
    \begin{split}
        T_{i,j}: \Lambda \to & \mathbb R \\
        \lambda \mapsto & T_{i,j}(\lambda)\\
        T'_{i,j}: \Lambda \to & \mathbb R \\
        \lambda \mapsto & T'_{i,j}(\lambda),
    \end{split}
\end{equation}
respectively. Since this notation will become cumbersome, we will
introduce a simplification. Let ${{\{b_i\}}_{1}}^{2N}=\{1,1,2,2,\ldots,N,N\}$ and
$a_i$ rotated one step so that ${{\{a_i\}}_{1}}^{2N}=\{1,2,2,\ldots,N,N,1\}$. Then
${{\{(a_i,b_i)\}}_{1}}^{2N}=\{(1,1)$,$(2,1)$,$(2,2)$,$(3,2)$,$\ldots$,
$(N,N)$,$(1,N)\}$. This allows us to
define subsets of $\Lambda$ as the sets on which Alice's and Bob's measurement
settings give coincident outcomes. For $1\leq i \leq 2N$ we have
\begin{equation}
    \Lambda_i \defeq \{\lambda: |T_{a_i,b_i}(\lambda)-T'_{a_i,b_i}(\lambda)| <
    \Delta T \}.
\end{equation}
We can now calculate the probability of coincidence as
\begin{equation}
    \gamma_N \defeq \inf_i P \left(\Lambda_i\right).
\end{equation}
Finally, for $1\leq i\leq 2N$ we have the conditional expectation defined as
\begin{equation}
    \label{eqn:condition}
    E(X_i|\Lambda_i) \defeq
    \int_{\Lambda_i}X_i(\lambda)dP(\lambda)
\end{equation}
where we use the convenient shorthand
\begin{equation}\label{eqn:shorthand}
    X_{i}\defeq
    A_{a_i}B_{b_i}
\end{equation}
for the product of the outcomes of Alice and Bob.

\section{The PBC inequality with coincidence
    probability}
We can now re-state \cref{thm:pbc} in terms of coincidence probability.
\begin{theorem}[PBC with coincidence probability]\label{thm:main}
    Let $N$ be an integer $\geq 2$ and $i$, $j$, and $k$ be integers between $1$
    and $2N$, and assume the prerequisites $(i)-(iii)$ of
    \cref{thm:pbc} hold almost everywhere together with
    \begin{enumerate}[(i)]
            \setcounter{enumi}{3}
        \item  Coincident events: Correlations are obtained on
            $\Lambda_i \subset \Lambda$.
    \end{enumerate}
    Then by defining
    \begin{equation}\label{eqn:pbc-coincident}
        \begin{split}
            S_{C,N}\defeq \big|E(X_1|\Lambda_1)+E(X_2|\Lambda_2) \big|
            +\cdots \\ +
            \big|E(X_{2N-1}|\Lambda_{2N-1})-E(X_{2N}|\Lambda_{2N})\big|
        \end{split}
    \end{equation}
    we get
    \begin{equation}\label[inequality]{eqn:main}
        S_{C,N}\leq \frac{4N-2}{\gamma_N} -2N
    \end{equation}
\end{theorem}
The remainder of this section is dedicated to proving this result. Note that
while the proof of \cref{thm:pbc} consists of adding expectation
values, this cannot be done for \cref{thm:main} since $\Lambda_i\neq \Lambda_j$
in general. Again, the ensemble changes with Alice's and Bob's measurement
settings, so the ensemble that \cref{thm:pbc} implicitly acts upon
is really
\begin{equation}
    \Lambda_I\defeq \bigcap_{i=1}^{2N} \Lambda_j,
\end{equation}
i.e.\ the intersection of all coincident subspaces of $\Lambda$. In other words,
prerequisites $(i)-(iii)$ yield
\begin{equation}\label{eqn:pbc-lambda-i}
    \begin{split}
        & \big|E(A_1B_1|\Lambda_I)+E(A_2B_1|\Lambda_I)
        \big|+\cdots \\ +
        & \big|E(A_{N}B_{N}|\Lambda_I)-E(A_{1}B_{N}|\Lambda_I)\big|
        \leq 2N-2
    \end{split}
\end{equation}
which again is a more precise re-statement of \cref{thm:pbc} where
we stress the conditional part. An experiment, however, will give us results on
the form $E(X_i|\Lambda_i)$, i.e. $E(X_i|\Lambda_I)$ is unavailable to the
experimenter. We therefore need to bridge the gap between experimental data and
\Cref{thm:main}, so we define following quantity which will act as a stepping
stone:
\begin{equation}\label{eqn:delta}
    \delta\defeq \inf_i \frac{P\Big(\bigcap_{j=1}^{2N} \Lambda_j \Big)}
    {P(\Lambda_i)
    } = \inf_i P\bigg( \bigcap_{j\neq i} \Lambda_j \bigg\vert \Lambda_i \bigg)
\end{equation}
Note that it is possible for the ensemble $\Lambda_I$ to be empty, but only when
$\delta=0$ and then \cref{eqn:main} is trivial. We can therefore assume
$\delta>0$ for the rest of the proof and our goal now is to give a lower bound
to $\delta$ in terms of the coincidence probability $\gamma_N$. We fix $i$ and
apply Boole's inequality:
\begin{equation}\label[inequality]{eqn:single-p-terms}
    P\bigg(\bigcap_{j\neq i} \Lambda_j \bigg\vert \Lambda_i \bigg)
    \geq2N-2+
    \sum_{j\neq i}
    P(\Lambda_j|\Lambda_i)
\end{equation}
and rewrite the summation terms:
\begin{equation}
    \label[inequality]{eqn:before-gamma}
    \begin{split}
        & P\left(\Lambda_j\middle \vert \Lambda_i\right)=
        \frac{P\left(\Lambda_i\cap\Lambda_j\right)}
        {P\left(\Lambda_i\right)}\\
        & = \frac{ P\left(\Lambda_i\right)+
            P\left(\Lambda_j\right)-
        P\left(\Lambda_i\cup\Lambda_j\right)}
        {P\left(\Lambda_i\right)} \\
        & \geq 1+\frac{P(\Lambda_j)-1}{P(\Lambda_i)}\geq
         1+\frac{\gamma_N-1}{\gamma_N} \\
         & =2-\frac{1}{\gamma_N}.
    \end{split}
\end{equation}
Inserting \cref{eqn:before-gamma} into \cref{eqn:single-p-terms} we get
\begin{equation}
    \label{eqn:p-estimate}
    \begin{split}
        &P\bigg(\bigcap_{j\neq i} \Lambda_j \bigg \vert \Lambda_i \bigg)
        \\
        &\geq
        2N-2+
        (2N-1)\left( 2-\frac{1}{\gamma_N} \right)\\
        & =  2N-\frac{2N-1}{\gamma_N},
    \end{split}
\end{equation}
and as \cref{eqn:before-gamma} is independent of $i$, inserting into
\cref{eqn:delta} gives
\begin{equation}
    \delta \geq 2N-\frac{2N-1}{\gamma_N},
    \label[inequality]{eqn:delta-in-gamma}
\end{equation}
and this is the desired lower bound. We now bound $S_{C,N}$ from above by adding
and subtracting $\delta E(X_i|\Lambda_I)$ in every term before applying the
triangle inequality and use \cref{eqn:pbc-lambda-i}:
\begin{equation}\label[inequality]{eqn:add-subtract}
    \begin{split}
        &S_{C,N}=
         \Big|E(X_1|\Lambda_1)-
        \delta E(X_1|\Lambda_I)\\
        &+\delta E(X_1|\Lambda_I)+
        E(X_2|\Lambda_2) \\
        &- \delta E(X_2|\Lambda_I)+
        \delta E(X_2|\Lambda_I)\Big| + \cdots\\
        &+ \Big|E(X_{2N-1}|\Lambda_{2N-1})-
        \delta E(X_{2N-1}|\Lambda_I)\\
        &+ \delta E(X_{2N-1}|\Lambda_I)-
        E(X_{2N}|\Lambda_{2N}) \\
        &+ \delta E(X_{2N}|\Lambda_I)-
        \delta E(X_{2N}|\Lambda_I)	\Big|
        \\
         \leq & \delta \Big (\big |E(X_1|\Lambda_I)+
            E(X_2|\Lambda_I)	| +  \cdots\\
            &+    |E(X_{2N-1}|\Lambda_I)-
        E(X_{2N}|\Lambda_I)	\big|\Big) \\
        &+  \sum_{i=1}^{2N}
        \big|E(X_i|\Lambda_i)-
        \delta E(X_i|\Lambda_I)\big| \\
         \leq & \delta  S_N
        +\sum_{i=1}^{2N}
        \big|E(X_i|\Lambda_i)-
        \delta E(X_i|\Lambda_I)\big|
    \end{split}
\end{equation}
To give an upper bound to the last sum, we need the following lemma:
\begin{lemma}\label{lemma:one-minus-delta}
    For $1\leq i \leq 2N$ and $0\leq \delta \leq 1$ we have the following
    inequality:
    \begin{align}
        \big|E(X_{i}|\Lambda_i)-
        \delta E(X_i|\Lambda_I)\big|\leq 1-\delta
    \end{align}
\end{lemma}
\begin{proof}
    It is clear that $\Lambda_I\subset \Lambda_i$. We can therefore split
    $\Lambda_i$ in two disjoint sets: $\Lambda_*\defeq \Lambda_i\setminus
    \Lambda_I$ and $\Lambda_I$. It follows that $\Lambda_I\cup
    \Lambda_*=\Lambda_i$ and we have
    \begin{equation}
        \begin{split}
            \big| & E(X_i|\Lambda_i)-
            \delta E(X_i|\Lambda_I)\big| \\
            \leq &  \big| P(\Lambda_*|\Lambda_i)
            E(X_i|\Lambda_{*})\big| \\
            & + \big| P(\Lambda_I|\Lambda_i)
            E(X_i|\Lambda_I) -\delta E(X_i|\Lambda_I)\big|  \\
            \leq &   P(\Lambda_*|\Lambda_i)
            E(|X_i|\big|\Lambda_*) \\
            &+ \Big(P(\Lambda_I|\Lambda_i)-\delta\Big)
            E(|X_i|\big|\Lambda_I ) \\
            \leq  & P(\Lambda_*|\Lambda_i) +
            P(\Lambda_I|\Lambda_i)-\delta = 1-\delta
        \end{split}
    \end{equation}
\end{proof}
\Cref{lemma:one-minus-delta} gives us
\begin{equation}\label[inequality]{eqn:sum-delta}
    \begin{split}
        \sum_{i=1}^{2N}\big|E(X_i|\Lambda_i)-
        \delta E(X_i|\Lambda_I)\big|  \leq2N(1-\delta)
    \end{split}
\end{equation}
The final step is to use
\cref{eqn:sum-delta,eqn:pbc,eqn:delta-in-gamma} on
\cref{eqn:add-subtract} which proves the desired result.

\section{Minimum coincidence probability}
The right-hand-side of \cref{eqn:main} increases as $\gamma_N$ goes down so
there exists a unique $\gamma_N$ so the bound on $S_{C,N}$ coincides with the
quantum-mechanical prediction in \cref{eqn:pbc-qm}. We define this
\emph{critical} coincidence probability as $\gamma_{\text{crit},N}$ and find it by
solving the following equation:
\begin{equation}
    2N\cos\left(\frac{\pi}{2N}\right)= \frac{4N-2}{\gamma_{\text{crit},N}} -2N
\end{equation}
and get
\begin{equation}\label{eqn:crit}
    \gamma_{\text{crit},N}=
    \frac{2N-1}{2N}\left(1+\tan^2\left( \frac{\pi}{4N} \right)\right).
\end{equation}
What remains to show is that
for all $\gamma_N\leq \gamma_{\text{crit},N}$ there exists a
local hidden variable (LHV) model that produces a $S_{C,N}$ that mimics
the predictions of quantum theory. Formally, we have the following theorem:
\begin{theorem}\label{thm:sufficiency}
    Let $N$ be an integer $\geq 2$. For every
    $\gamma_N\leq\gamma_{\text{crit},N}$ is is possible to construct an LHV
    model fulfilling the prerequisites (i) --- (iv) of \cref{eqn:main} so that
    \begin{equation}
        S_{C,N}=2N\cos\left(\frac{\pi}{2N}\right).
    \end{equation}
\end{theorem}
\begin{figure}
    \centering
    \includegraphics[width=\linewidth]{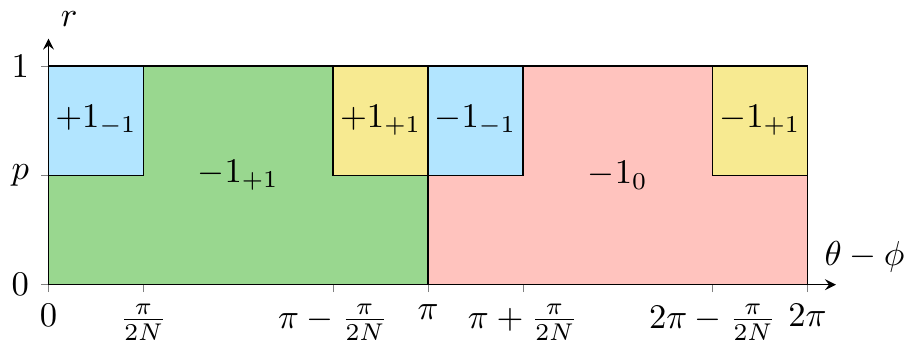}
    \caption{LHV model that gives the outcomes for Alice's and Bob's detectors.\label{fig:lhv}}
\end{figure}
We explicitly prove \Cref{thm:sufficiency} by constructing the LHV model
depicted in
\cref{fig:lhv}. Here, the hidden variable is on the form $(r,\theta)$ and
uniformly distributed over $0\leq r\leq 1$ and $0\leq \theta\leq 2\pi$.
The LHV model defines the random variables $A_i$ and $B_i$ and arrival times
$T_i$ and $T'_i$, where we adapt the shorthand from \cref{eqn:shorthand} to the
definition in \cref{eqn:time}. We choose $\phi$ to be a function of $i$ in the
following way for Alice's detector:
\begin{equation}
    \phi(i)\defeq
    a_i\frac{\pi}{2N}
    \label{eqn:alice-phi}
\end{equation}
and the following way for Bob's detector:
\begin{equation}
    \phi(i)\defeq
    b_i\frac{\pi}{2N}.
    \label{eqn:bob-phi}
\end{equation}
In \cref{fig:lhv}, $\phi$ acts as a shift in the $\theta$ direction (with
wraparound when neccessary). The case $i=1$ is depicted in
\cref{fig:lhv-overlap}, and
by choosing $\Delta T=3/2$ we get coincidence for a time difference of 0 and 1
units (solid background), and non-coincidence for a time difference of two units
(cross-hatched background).
\begin{figure}
    \centering
    \includegraphics[width=\linewidth]{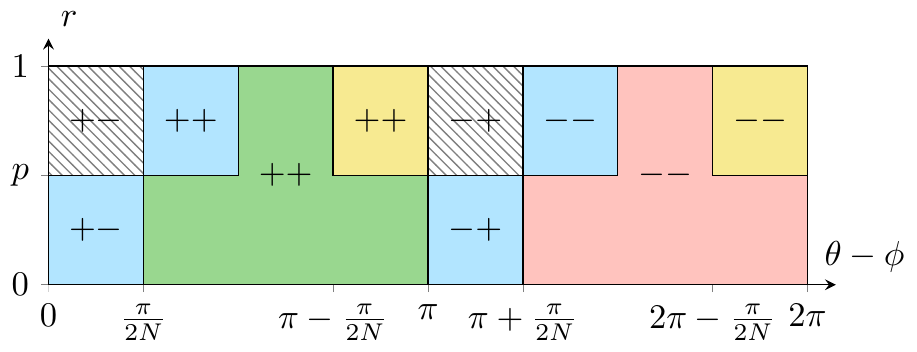}
    \caption{Alice's and Bob's outcome patterns for the case $i=1$. The two
        plus/minus signs show Alice's and Bob's outcome, respectively. The
    cross-hatch areas show outcomes that are non-coincident given $\Delta
T=3/2$.\label{fig:lhv-overlap}}
\end{figure}
We compute the probability of coincidence in \cref{fig:lhv-overlap},
$P(\Lambda_i)=(2N-1+p)/2N$ and find that it is independent of $i$. Therefore,
\begin{equation}\label{eqn:gamma-p}
    \gamma_N=(2N-1+p)/2N.
\end{equation}
In addition, for $1\leq i \leq 2N-1$,
\begin{equation}
    \begin{split}
        P(X_i=+1|\Lambda_i)=& \frac{P(X_i=+1)}{P(\Lambda_i)}\\
        = & \frac{2N-1}{2N-1+p}
    \end{split}
\end{equation}
and
\begin{equation}
    \begin{split}
        P(X_i=-1|\Lambda_i)=& \frac{P(X_i=+1\cap
        \Lambda_i)}{P(\Lambda_i)}\\
        =& \frac{p}{2N-1+p}
    \end{split}
\end{equation}
which gives
\begin{equation}
    \begin{split}
        & E(X_i|\Lambda_i)\\
        & =P(X_i=+1|\Lambda_i)-P(X_i=-1|\Lambda_i)\\
        & =\frac{2N-1-p}{2N-1+p}
    \end{split}
\end{equation}
for $1\leq i\leq 2N-1$. A similar calculation yields
\begin{equation}
    E(X_{2N}|\Lambda_{2N})=-\frac{2N-1-p}{2N-1+p}.
\end{equation}
We now insert the predictions of the LHV model into
\cref{eqn:pbc-coincident} to get
\begin{equation}
    \begin{split}
        S_{\text{LHV},N}\defeq \big|E(X_1|\Lambda_1)+E(X_2|\Lambda_2) \big|
        +\cdots\\ +
        \big|E(X_{2N-1}|\Lambda_{2N-1})-E(X_{2N}|\Lambda_{2N})\big|.
    \end{split}
\end{equation}
As we want the LHV model to mimic the predictions of quantum mechanics (from
\cref{eqn:pbc-qm}) we put
\begin{equation}
    S_{\text{LHV},N}= 2N\cos \left( \frac{\pi}{2N} \right)
\end{equation}
which gives
\begin{equation}
    \frac{2N-1-p}{2N-1+p}= \cos \left( \frac{\pi}{2N} \right).
\end{equation}
Solving for $p$ we get
\begin{equation}
    p=(2N-1)\tan^2\left(\frac{\pi}{4N}\right)
\end{equation}
and \cref{eqn:gamma-p} then gives us
\begin{equation}
    \gamma_N=\frac{2N-1}{2N}\left(1+\tan^2\left( \frac{\pi}{4N} \right)\right)
\end{equation}
which coincides with $\gamma_{\text{crit},N}$. The model in \cref{fig:lhv}
is a constructive proof of \cref{thm:sufficiency} as it produces the same
output statistics as quantum mechanics with coincidence probability
$\gamma_{\text{crit},N}$. We finally note that it is trivial to modify the LHV
model to give any $\gamma \leq \gamma_{\text{crit},N}$ which finishes the proof.

The LHV model in \cref{fig:lhv} mimics almost every statistical property of a
truly quantum-mechanical experiment (see~\cite{Larsson2004Bell}) and shows it is
possible to fake a violation of the PBC inequality if the coincidence
probability is lower than the critical value. It is therefore important that any
experiment relying on a PBC inequality violation takes the coincidence
probability into account before ruling out a classical model.

As the number of measurement settings $N$ goes to infinity the critical coincidence
probability $\gamma_{\text{crit},N}$ goes to 1. Therefore, achieving the required
coincidence becomes harder as more measurement settings are used. If we define
$\eta_{\text{crit},N}$ as the
minimum required \emph{detection efficiency} for a violation of the PBC inequality free
of the detection loophole (see~\cite{Cabello2009Minimum} for full details)
we get
\begin{equation}\label{eqn:eta}
    \eta_{\text{crit},N}=
    \frac{2}{\frac{N}{N-1}\cos\left(\frac{\pi}{2N}\right)+1}
\end{equation}
and note that $\gamma_{\text{crit},N}>\eta_{\text{crit},N}$ for all $N\geq 2$.
In addition, the critical coincidence probability for the special CHSH case
$N=2$ is \SI{87.87}{\percent} which agrees with previous
works~\cite{Larsson2004Bell}. See \cref{tab:critical} for critical probabilities
for the cases $N=2,3,4,5$ and note that both $\gamma_{\text{crit},N}$ and
$\eta_{\text{crit},N}$ are strictly increasing in $N$. Note that a loophole-free
experiment requires \emph{both} the coincidence probability and detection
efficiency be in excess of their respective thresholds.

\begin{table}
    \centering
    \begin{tabular}{cll}
        $N$ & $\gamma_{crit,N}$ & $\eta_{crit,N}$\\
        \midrule
        2 (CHSH) & \SI{87.87}{\percent} & \SI{82.84}{\percent}\\
        3 & \SI{89.32}{\percent} & \SI{86.99}{\percent}\\
        4 & \SI{90.96}{\percent}& \SI{89.61}{\percent}\\
        5 & \SI{92.26}{\percent}& \SI{91.37}{\percent} \\
        $N$	& \multicolumn{2}{c}{Increases with $N$}
    \end{tabular}
    \caption{Critical coincidence probabilities $\gamma_{crit,N}$ and detection
        probabilities $\eta_{crit,N}$ for a loophole-free violation of the PBC
        equality for 2,3,4, and 5 measurement settings. Note that $N=2$
    corresponds to the special CHSH case.\label{tab:critical}}
\end{table}

While reaching $\gamma_{\text{crit},N}$ is less challenging for small $N$, some
applications do requires a PBC inequality with a higher number of settings. An
example is the Franson interferometer~\cite{Franson1989Bell}, where
postselection leads to a loophole for $N=2$ but not for $N\geq
3$~\cite{Jogenfors2014Energy}. In fact,
$N=5$ is optimal for that setup in terms of violation, however
\cref{tab:critical} shows that the corresponding minimal coincidence probability
is as high as \SI{92.26}{\percent}, which is a considerable challenge.

\section{Conclusion}

The PBC inequality is a powerful tool for testing local realism in applications
where the CHSH test is insufficient. We have found the minimum required
coincidence probability for a violation of the PBC inequality without the
fair-coincidence assumption. This bound is tight, so any application of the PBC
inequality that relies on a violation of local realism must have at least this
coincidence probability, unless the perilous fair-coincidence assumption is to
be made. If not, and if the coincidence probability is below the critical
threshold, an attacker can construct a local realist model from which all
measurements can be predicted.

\printbibliography{}
\end{document}